\documentclass{article}
\usepackage[letterpaper, hmargin=1in, vmargin={1in, 1.4in}]{geometry}
\usepackage{titling}
\usepackage{amsmath,amssymb,amsthm}
\usepackage{xspace}
\usepackage{multirow}
\usepackage{makecell}
\usepackage{bm}
\usepackage{algorithm}
\usepackage[noend]{algpseudocode}
\usepackage{url}
\usepackage{tikz}
\usepackage[braket]{qcircuit}
\usepackage{calc}

\allowdisplaybreaks[4]

\newtheorem{theorem}{Theorem}
\newtheorem{lemma}{Lemma}
\newtheorem{proposition}[lemma]{Proposition}
\newtheorem{corollary}[lemma]{Corollary}
\newtheorem{fact}[lemma]{Fact}
\theoremstyle{definition}

\newtheorem{ralgo}{Reduction algorithm}

\theoremstyle{remark}

\title{Exponential-time Quantum Algorithms for Graph Coloring Problems}
\author{Kazuya Shimizu$^\dagger$ and Ryuhei Mori${}^\dagger{}^\ddagger$}
\date{\small \texttt{shimizu.k.ap@m.titech.ac.jp}\\ \texttt{mori@c.titech.ac.jp}\\ $\dagger$School of Computing, Tokyo Institute of Technology, Japan\\ $\ddagger$Japan Science and Technology Agency, PRESTO}

\begin{document}
\begin{titlingpage}
\maketitle

\begin{abstract}
The fastest known classical algorithm deciding the $k$-colorability of $n$-vertex graph requires running time $\Omega(2^n)$ for $k\ge 5$.
In this work, we present an exponential-space quantum algorithm computing the chromatic number with running time $O(1.9140^n)$ using quantum random access memory (QRAM)\@.
Our approach is based on Ambainis et al's quantum dynamic programming with applications of Grover's search to branching algorithms.
We also present a polynomial-space quantum algorithm not using QRAM for the graph $20$-coloring problem with running time $O(1.9575^n)$.
In the polynomial-space quantum algorithm, we essentially show $(4-\epsilon)^n$-time \textit{classical} algorithms that can be improved quadratically by Grover's search.
\end{abstract}
\end{titlingpage}
\setcounter{page}{1}

\section{Introduction}
Exhaustive search is believed to be (almost) the fastest classical algorithm for many NP-complete problems including SAT, hitting set problem, etc~\cite{cygan2016problems}.
Grover's quantum search quadratically improves the running time of exhaustive search~\cite{grover1996fast}.
Hence, the best classical running time for many NP-complete problems can be quadratically improved by quantum algorithms.
On the other hand, non-trivial faster classical algorithms are known for some NP-complete problems including the travelling salesman problem (TSP), the graph coloring problem, etc.
For these problems, more complicated techniques, such as dynamic programming, arithmetic algorithm based on inclusion--exclusion principle, etc., are used in the fastest known classical algorithms.
It is not obvious how to boost these classical algorithms by a quantum computer.
Recently, Ambainis et~al.\ showed a general idea of quantum dynamic programming using quantum random access memory (QRAM) and showed quantum speedup for many NP-hard problems including TSP, set cover, etc~\cite{ambainis2019quantum}.
Ambainis et al.'s work gives a new general method for exact exponential-time quantum algorithms.

In this work, we present exact exponential-time quantum algorithms for the graph coloring problem.
The fastest known classical algorithm computes the chromatic number of $n$-vertex graph with running time $\mathrm{poly}(n)2^n$ on the random access memory (RAM) model.
The main result of this work is the following theorem.
\begin{theorem}\label{thm:main}
There is an exponential-space bounded error quantum algorithm using QRAM for the chromatic number problem with running time $O^*\left((2^{37/35}3^{3/7}5^{-9/70}7^{-5/28})^n\right)\footnote{In this paper, $O^*(f(n))$ means $O(\mathrm{poly}(n)f(n))$.}=O(1.9140^n)$.
\end{theorem}
The quantum algorithm in Theorem~\ref{thm:main} is based on Ambainis et al's quantum dynamic programming for TSP with applications of Grover's search to Byskov's algorithm enumerating all maximal independent sets (MISs) of fixed size~\cite{byskov2004enumerating}.
Byskov's algorithm is not naive exhaustive search, but is a branching algorithm (also referred as Branch \& Reduce), for which Grover's search can be applied~\cite{furer2008solving}.
While RAM is widely accepted model of classical computation, QRAM is sometimes criticized due to the difficulty of implementation.
In this paper, we also present quantum algorithms not using QRAM\@.
\begin{theorem}\label{thm:maing}
For $k\le 20$, there exists $\epsilon>0$ such that
there is polynomial-space bounded error quantum algorithms not using QRAM for the graph $k$-coloring problem with running time $2^{(1-\epsilon)n}$.
\end{theorem}
Note that classical algorithms with running time $2^{(1-\epsilon)n}$ are known only for $k=3,4$~\cite{beigel20053}, \cite{byskov2004enumerating}.
Running times of the quantum algorithms in Theorem~\ref{thm:maing} are shown in Table~\ref{table:20}.
For proving Theorem~\ref{thm:maing}, we essentially show \textit{classical} algorithms with running time $4^{(1-\epsilon)n}$ that can be improved quadratically by Grover's search.
These classical algorithms are obtained by generalizing Byskov's techniques for reducing the graph $k(\ge 4)$-coloring problem to the graph 3-coloring problem~\cite{byskov2004enumerating}.

\begin{table}[t]
\centering
\caption{$O(2^{d^*_kn})$-time quantum algorithms not using QRAM\@.}
\label{table:20}
\begin{tabular}{|c|c|c|}
\hline
$k$& $d^*_k$& $2^{d^*_k}$\\
\hline
3 & 0.2051 & 1.1528\\
4 & 0.4039 & 1.3231\\
5 & 0.5553 & 1.4695\\
6 & 0.6099 & 1.5261\\
7 & 0.7234 & 1.6511\\
8 & 0.7299 & 1.6585\\
\hline
\end{tabular}
\begin{tabular}{|c|c|c|}
\hline
$k$& $d^*_k$& $2^{d^*_k}$\\
\hline
9 & 0.8041 & 1.7460\\
10 & 0.8298 & 1.7775\\
11 & 0.8298 & 1.7775\\
12 & 0.8676 & 1.8246\\
13 & 0.8874 & 1.8499\\
14 & 0.8938 & 1.8580\\
\hline
\end{tabular}
\begin{tabular}{|c|c|c|}
\hline
$k$& $d^*_k$& $2^{d^*_k}$\\
\hline
15 & 0.9488 & 1.9303\\
16 & 0.9488 & 1.9303\\
17 & 0.9488 & 1.9303\\
18 & 0.9536 & 1.9366\\
19 & 0.9690 & 1.9575\\
20 & 0.9691 & 1.9575\\
\hline
\end{tabular}
\end{table}  
             
\subsection{Related work}
Since a graph is $k$-colorable if and only if the set of vertices can be partitioned into $k$ independent sets, many algorithms for the graph $k$-coloring problem use enumeration algorithms of independent sets.
There is a simple branching algorithm enumerating all MISs in time $O^*(3^{n/3})=O(1.4423^n)$~\cite{fomin2010}.
Lawler showed that 3-colorability can be decided in time $O^*(3^{n/3})$ by enumerating all MISs and checking the bipartiteness of the subgraph induced by the complement of each MIS~\cite{lawler1976note}.
Lawler also showed that the chromatic number can be computed in time $O(2.4423^n)$ by a simple dynamic programming.

Beigel and Eppstein showed an efficient algorithm for the graph 3-coloring problem with running time $O(1.3289^n)$~\cite{beigel20053}.
Byskov showed reduction algorithms from the graph $k(\ge 4)$-coloring problem to the graph 3-coloring problem~\cite{byskov2004enumerating}.
By using Beigel and Eppstein's graph 3-coloring algorithm, Byskov showed classical algorithms for the graph 4-, 5- and 6-coloring problems with running time $O(1.7504^n)$, $O(2.1592^n)$ and $O(2.3289^n)$, respectively.
Fomin et al.\ showed an algorithm for the graph 4-coloring problem with running time $O(1.7272^n)$ by using the path decomposition~\cite{fomin2007improved}.

In 2006, Bj\"orklund and Husfeldt, and Koivisto showed an exponential-space $O^*(2^n)$-time algorithm for the chromatic number problem on the RAM model~\cite{bjorklund2006inclusion}, \cite{koivisto20062}.
These algorithms are based on the inclusion--exclusion principle.
They also showed that if there is a polynomial-space $O^*(\alpha^n)$-time algorithm counting the number of independent sets, then there is a polynomial-space $O^*((1+\alpha)^n)$-time algorithm computing the chromatic number~\cite{bjorklund2006inclusion}, \cite{bjorklund2009set}.
Since the fastest known polynomial-space algorithm computes the number of independent sets with running time $O(1.2356^n)$~\cite{gaspers2017faster}, there is a polynomial-space $O(2.2356^n)$-time algorithm computing the chromatic number.

There is almost no previous theoretical work on quantum algorithms for the graph coloring problems.
F\"urer mentioned that Grover's algorithm can be applied to branching algorithms so that Beigel and Eppstein's algorithm for the graph 3-coloring problem can be improved to running time $O(\sqrt{1.3289}^n)=O(1.1528^n)$~\cite{furer2008solving}.
The quantum algorithms for Theorem~\ref{thm:maing} are basically obtained by applying Grover's search to generalized Byskov's reduction algorithms on the basis of F\"urer's observation.

For general NP-hard problems, Ambainis et al.\ showed exponential-space exponential-time quantum algorithms using QRAM for many NP-hard problems~\cite{ambainis2019quantum}.
The quantum algorithm for Theorem~\ref{thm:main} is based on Ambainis et al's algorithm for TSP with application of Grover's search to Byskov's algorithm enumerating MISs of fixed size on the basis of F\"urer's observation.

\subsection{Overview of quantum algorithms}
\subsubsection{Quantum algorithm for the chromatic number problem}\label{subsubsec:algqram}
Similarly to Ambainis et al.'s quantum algorithm for TSP, the quantum algorithm for Theorem~\ref{thm:main} is a simple divide-and-conquer algorithm with dynamic programming approach.
The basic classical algorithm was shown in~\cite[Proposition 3]{bjorklund2008exact}.
The chromatic number of a graph $G$ is equal to a sum of the chromatic numbers of $G[S]$ and $G[V\setminus S]$ for some non-empty $S\subsetneq V$ unless $G$ is one-colorable.
If we can assume that $S$ has size exactly $\lfloor n/2\rfloor$ or $\lceil n/2\rceil$, then we can consider a classical algorithm that recursively finds $S$ of size $\lfloor n/2\rfloor$ or $\lceil n/2\rceil$ minimizing $\chi(G[S])+\chi(G[V\setminus S])$.
Let $T(n)$ be the running time of this algorithm.
Then, it follows $T(n) = \binom{n}{\lfloor n/2\rfloor}(T(\lfloor n/2\rfloor)+T(\lceil n/2\rceil))$
so that we can apply Ambainis et al's quantum dynamic programming straightforwardly and obtain $O(1.7274^n)$-time quantum algorithm~\cite{ambainis2019quantum}.
However, the balanced partition $S$ satisfying $\chi(G)=\chi(G[S])+\chi(G[V\setminus S])$ does not necessarily exist.
Hence, we use the following useful fact.
\begin{fact}\label{fa:bcut}
Let $a_1,\dotsc,a_k$ be positive integers, and $n:=\sum_{i=1}^ka_i$.
Assume that $a_1\ge a_i$ for all $i\in\{1,2,\dotsc,k\}$.
Then, for any $m\in\{1,2,\dotsc,n-1\}$, there exists $S\subseteq\{2,3,\dotsc,k\}$ such that
$\sum_{i\in S}a_i\le m$ and $\sum_{i\in \{2,\dotsc,k\}\setminus S} a_i\le n-m-1$.
\end{fact}
\begin{proof}
Let $t:=\max\bigl\{j\in\{2,\dotsc,k\}\mid \sum_{i=2}^j\le m\bigr\}$. Let $S:=\{2,3,\dotsc,t\}$.
Then, $\sum_{i\in\{t+1,t+2,\dotsc,k\}} a_i \le \sum_{i\in\{1,t+2,t+3,\dotsc,k\}} a_i = n - \sum_{i\in\{2,3,\dotsc,t+1\}} a_i \le n - m + 1$.
\end{proof}
From Fact~\ref{fa:bcut}, we can consider following quantum algorithm computing the chromatic number.
First, the algorithm precomputes the chromatic number of all induced subgraphs with size at most $n/4$.
This precomputation is based on Lawler's formula
\begin{equation}\label{eq:lawler}
\chi(G) = 1 + \min_{I\in\mathrm{MIS}(G)} \chi(G[V\setminus I])
\end{equation}
where $\mathrm{MIS}(G)$ denotes the set of all MISs of $G$~\cite{lawler1976note}.
There is a classical algorithm enumerating all MISs with running time $3^{n/3}$.
We will show in Section~\ref{sec:branch} that Grover's search can be applied to this algorithm so that the quantum algorithm can search all MISs with running time $3^{n/6}$.
Here, computed chromatic numbers are stored to QRAM.
Hence, we can apply Grover's search for computing the minimum in~\eqref{eq:lawler}.
The precomputation requires the running time $O^*\left(\sum_{i=1}^{n/4} \binom{n}{i}3^{i/6}\right)=O(1.8370^n)$.
Then, the main part of the algorithm computes the chromatic number of $G$ by using the formula
\begin{equation*}
\chi(G) = 1 + \min_{I\in\mathrm{MIS}(G)} \min_{S\subseteq V\setminus I,\, |S|\le n/2,\, |V\setminus I\setminus S|\le n/2}\left\{\chi(G[S]) + \chi(G[V\setminus I\setminus S])\right\}
\end{equation*}
for $\chi(G)\ge 3$.
This formula is justified by Fact~\ref{fa:bcut} for $m=\lfloor n/2\rfloor$.
At the third level, the number of vertices in a graph is at most $n/4$.
Hence, the chromatic number was precomputed and stored to QRAM.
The running time of the main part of the quantum algorithm is
\begin{equation*}
O^*\left(3^{n/6} \sqrt{\binom{n}{n/2}} 3^{n/12}\sqrt{\binom{n/2}{n/4}}\right)
=O(2.2134^n).
\end{equation*}
Quantum algorithm for Theorem~\ref{thm:main} searches all MISs of size $t$ for each $t\in\{1,2,\dotsc,n\}$ separately.
Then, precise analysis shows that the running time of the improved quantum algorithm is $O(1.9140^n)$.

\subsubsection{Quantum algorithms for the graph $k$-coloring problem}
We will derive classical algorithms that can be improved quadratically by Grover's search.
In the classical algorithms, the graph $k$-coloring problem is reduced to the graph $k'$-coloring problems for some $k'<k$.
Since a graph $G$ is $k$-colorable if and only if there exists a subset $S$ of vertices such that $G[S]$ is $\lfloor k/2\rfloor$-colorable and $G[V\setminus S]$ is $\lceil k/2\rceil$-colorable.
Let us consider a classical algorithm that simply searches $S\subseteq V$ satisfying the above condition.
Let $T_k(n)$ be the running time of this algorithm for the graph $k$-coloring problem.
Then, $T_k(n)$ satisfies
\begin{align*}
T_1(n) &= T_2(n) = 1,\\
T_k(n) &= \sum_{i=0}^n\binom{n}{i}(T_{\lfloor k/2\rfloor}(i) + T_{\lceil k/2\rceil}(n-i))
\end{align*}
where polynomial factors in $n$ are ignored.
Then, we obtain $T_4(n) = O^*(2^n)$, $T_8(n)= O^*(3^n)$ and $T_{16}(n)=O^*(4^n)$.
Let us consider a quantum algorithm that uses Grover's search for finding $S$.
Let $T_k^*(n)$ be the running time of the quantum algorithm.
Then, it follows $T_k^*(n) = \sum_{i=0}^n\sqrt{\binom{n}{i}}(T_{\lfloor k/2\rfloor}^*(i) + T_{\lceil k/2\rceil}^*(n-i))$,
which implies $T_k^*(n) = O^*(\sqrt{T_k(n)})$.
Hence, we obtain $T_4^*(n) = O(1.4143^n)$, $T_8^*(n)=O(1.7321^n)$ and $T_{16}^*(n)=O^*(2^n)$.
This yields a weaker version of Theorem~\ref{thm:maing} that is valid for $k\le 8$ rather than $k\le 20$.

\subsection{Organization}
In Section~\ref{sec:pre}, notations and known classical and quantum algorithms are introduced.
In Section~\ref{sec:branch}, we present details of quantum algorithm for branching algorithms.
In Section~\ref{sec:chr}, we prove Theorem~\ref{thm:main}.
In Section~\ref{sec:kcol}, we prove a weaker version of Theorem~\ref{thm:maing} that is valid for $k\le 19$ rather than $k\le 20$.
Theorem~\ref{thm:maing} is obtained by improving the quantum algorithms in Section~\ref{sec:kcol}.
The details of the proof of Theorem~\ref{thm:maing} are shown in Appendix~\ref{apx:20alg}.

\section{Preliminaries}\label{sec:pre}
\subsection{Definitions and notations}
For a finite vertex set $V$, a set $E$ of edges consists of subsets of $V$ of size two.
A pair $(V,E)$ of finite vertex set $V$ and a set $E$ of edges is called an undirected simple graph.
In this paper, we simply call a graph rather than an undirected simple graph.
The number of vertices $|V|$ is denoted by $n$.
A mapping $c\colon V \to \{1,2,\dotsc,k\}$ is called $k$-coloring if $c(v)\ne c(w)$ for all $\{v,w\}\in E$.
For a graph $G$, the smallest $k$ such that there exists a $k$-coloring is called the chromatic number of $G$, and denoted by $\chi(G)$.
A subset $I\subseteq V$ of vertices is called an independent set if $\{v,w\}\notin E$ for all $v, w\in I$.
An independent set $I$ is said to be maximal if there is no strict superset of $I$ that is an independent set.
A maximal independent set of size $t$ is called $t$-MIS\@.
For $S\subseteq V$, $G[S]$ denotes a induced subgraph $(S, \{\{v,w\}\in E\mid v, w\in S\})$ of $G$.
Let $h(\delta) := -\delta\log\delta - (1-\delta)\log(1-\delta)$ for $\delta\in[0,1]$ where $0\log 0 = 0$.
In this paper, the base of logarithm is 2.
The notation $g(n) = O^*(f(n))$ means that $g(n) = O(n^cf(n))$ for some constant $c$.
For $O^*(\lambda^n)$, we often round $\lambda$ up to the fourth digit after the decimal point. 
In this case, we can use $O()$ rather than $O^*()$.
For example, we often write $g(n) = O(1.4143^n)$ rather than $g(n) = O^*(2^{n/2})$.
The notation $g(n) = \widetilde{O}(f(n))$ means that $g(n) = O((\log f(n))^c f(n))$ for some constant $c$.

\subsection{Known classical algorithm for enumerating all $t$-MISs}
Byskov showed the following theorem.
\begin{theorem}[Byskov~\cite{byskov2004enumerating}]\label{thm:byskov}
The maximum number of $t$-MISs of $n$-vertex graphs is
\begin{equation*}
I(n,t) := \lfloor n/t\rfloor^{(\lfloor n/t\rfloor +1) t - n}(\lfloor n/t\rfloor +1)^{n-\lfloor n/t\rfloor t}.
\end{equation*}
Furthermore, there is a classical algorithm enumerating all $t$-MISs of $n$-vertex graph in time $O^*(I(n,t))$.
\end{theorem}
We can straightforwardly obtain the following lemma and corollary.
\begin{lemma}\label{lem:E}
For any constant $\delta\in(0,1)$, $I(n, \lfloor\delta n\rfloor)=O(2^{E(\delta)n})$ where
\begin{equation*}
E(\delta) := ((\lfloor \delta^{-1}\rfloor +1)\delta - 1)\log\lfloor\delta^{-1}\rfloor + (1-\lfloor\delta^{-1}\rfloor\delta)\log(\lfloor\delta^{-1}\rfloor+1).
\end{equation*}
Here, $E(\delta)$ is concave (and hence, continuous) and piecewise linear for $\delta\in(0,1)$.
The maximum of $E(\delta)$ is given at $\delta=1/3$.
\end{lemma}
\begin{corollary}\label{cor:I}
For any $a\in\mathbb{R}_{\ge 0}$ and $t\in\mathbb{Z}_{\ge 3}$, the maximum of $E(\delta) -  a \delta$ for $\delta\in[1/t,1]$ is given at $\delta=1/s$ for some $s\in\{3,4,\dotsc,t\}$.
\end{corollary}

\subsection{Grover's search}
Here, Grover's search is briefly introduced without using quantum circuit, unitary oracle, etc.
\begin{theorem}[Grover~\cite{grover1996fast}, Boyer et al.~\cite{boyer1998tight}]\label{thm:grover}
Let $A\colon \{1,2,\dotsc,N\}\to\{0,1\}$ be a bounded-error quantum algorithm with running time $T$.
Then, there is a bounded-error quantum algorithm computing $\bigvee_{x\in\{1,\dotsc,N\}} A(x)$ with running time $\widetilde{O}(\sqrt{N}T)$.
If it is guaranteed that $|A^{-1}(1)|\ge M$ or $|A^{-1}(1)|=0$, then there is a bounded-error quantum algorithm with running time $\widetilde{O}(\sqrt{N/M}T)$.
\end{theorem}

\begin{theorem}[D\"urr and H{\o}yer~\cite{durr1996quantum}]\label{thm:min}
Let $A\colon \{1,2,\dotsc,N\}\to\{1,2,\dotsc,M\}$ be a bounded-error quantum algorithm with running time $T$.
Then, there is a bounded-error quantum algorithm computing $\min_{x\in\{1,\dotsc,N\}} A(x)$ with running time $\widetilde{O}(\sqrt{N}T)$.
\end{theorem}

\subsection{QRAM}
QRAM is the quantum analogue of RAM which can be accessed in a superposition~\cite{giovannetti2008quantum}.
QRAM has been used in many quantum algorithms~\cite{ambainis2019quantum}.
RAM is the memory that can be accessed in constant or logarithmic time with respect to the memory size.
For computing the minimum of $f(x, W)$ for all $x\in\{1,2,\dotsc,N\}$ where $W$ denotes a read-only RAM,
we can replace RAM with QRAM and apply Grover's search for computing the minimum.
Then, we obtain $\widetilde{O}(\sqrt{N}T)$-time quantum algorithm where $T$ denotes the running time for computing $f$.

\section{Grover's search for branching algorithms}\label{sec:branch}
F\"urer mentioned that Grover's search can be applied to branching algorithms~\cite{furer2008solving}.
Since the details of the quantum algorithm were not explicitly described in~\cite{furer2008solving}, we will show the details in this section.
A branching algorithm is an algorithm which recursively reduce a problem into some problems of smaller parameters.
We now consider decision problems with $\ell$ parameters $n_1,n_2,\dotsc,n_\ell$ that are non-negative integers.
If the parameters are sufficiently small, we do not apply any branching rule and solve this problem in some way.
For a problem $P$ with parameters $n_1,\dotsc,n_\ell$ that are not sufficiently small, we choose a branching rule $b(P)$ such that $P$ is reduced to $m_{b(P)}$ problems $P_1,P_2\dotsc,P_{m_{b(P)}}$ of the same class.
Here, $P_i$ has parameters $f^{b(P),i}_1(n_1),\dotsc,f^{b(P),i}_\ell(n_\ell)$ for some function $f_j^{b(P),i}$ satisfying $f_j^{b(P),i}(n_j)\le n_j$ for $i=1,2,\dotsc,m_{b(P)}$ and $j=1,2,\dotsc,\ell$.
At least one of the parameters of $P_i$ must be smaller than the same parameter of $P$ for all $i\in1,2,\dotsc,m_{b(P)}$.
The solution of $P$ is true if and only if at least one of the solutions of $P_1,\dotsc,P_{m_{b(P)}}$ is true. Hence, we will call this algorithm OR-branching algorithm.
For a problem $P$ of this class, we can consider a \textit{computation tree} that represents the branchings of the reductions.
The computation tree for $P$ is a single node if $P$ has sufficiently small parameters so that any branching rule is not performed,
and is a rooted tree where children of the root node are the root nodes of the computation trees for $P_1,P_2,\dotsc,P_{m_{b(P)}}$ if some branching rule $b(P)$ is applied to $P$.
Let $L(n_1,\dotsc,n_\ell)$ be the maximum number of leaves of the computation tree for $P$ with parameters $n_1,\dotsc,n_\ell$.
Assume that the running time of the computation at a non-leaf node, including computations of $b(P)$, $P_i$, and $f_j^{b(P),i}$, is polynomial with respect to $n_1,\dotsc,n_\ell$.
Then, the total running time of the OR-branching algorithm is at most $\mathrm{poly}(n_1,\dotsc,n_\ell)L(n_1,\dotsc,n_\ell)T$
where $T$ is the running time for the computation at a leaf node.
We can apply Grover's search to OR-branching algorithms if we have an upper bound of $L(n_1,\dotsc,n_\ell)$ with some properties.
\begin{lemma}\label{lem:or}
Let $U(n_1,\dotsc,n_\ell)$ be an upper bound of $L(n_1,\dotsc,n_\ell)$ that can be computed in polynomial time with respect to the parameters, and satisfies
\begin{equation*}
U(n_1,\dotsc,n_\ell)\ge \sum_{i=1}^{m_{b}} U(f_1^{b,i}(n_1),\dotsc,f_\ell^{b,i}(n_\ell))
\end{equation*}
for any branching rule $b$.
Then, there is a bounded-error quantum algorithm with running time $\mathrm{poly}(n_1,\dotsc,n_\ell)\allowbreak\sqrt{U(n_1,\dotsc,n_\ell)}T$.
\end{lemma}
\begin{proof}
If we can assign an integer $s\in\{1,2,\dotsc,U(n_1,\dotsc,n_\ell)\}$ to every leaf of the computation tree, and can compute the corresponding leaf from given $s$ in polynomial time with respect to the parameters,
then, we can apply Grover's search for computing
\begin{equation*}
a(P) = \bigvee_{Q\in W(P)} a(Q)
\end{equation*}
where $a(P)$ denotes the solution of a problem $P$ and $W(P)$ denotes the set of all problems corresponding to leaves of the computation tree for $P$.
Then, we obtain quantum algorithm with running time $\mathrm{poly}(n_1,\dotsc,n_\ell)\sqrt{U(n_1,\dotsc,n_\ell)}T$~\cite{furer2008solving}.
The algorithm computing $s$-th leaf of a problem $P$ is shown in Algorithm~\ref{alg:leaf}.
We will show the validity of Algorithm~\ref{alg:leaf}.
\begin{algorithm}[t]
\begin{algorithmic}[1]
\Function{Leaf}{$P$, $s$}
\If {$P$ is a leaf} \Return $P$
\EndIf
\State Compute the branching rule $b\gets b(P)$
\For {$i\in\{1,2,\dotsc,m_{b}-1\}$}
\State Compute $P_i$ and its parameters $n'_1,\dotsc,n'_\ell = f_1^{b,i}(n_1),\dotsc,f_\ell^{b,i}(n_\ell)$
\If {$s\le U(n'_1,\dotsc,n'_\ell)$}
\Return $\textsc{Leaf}(P_i, s)$
\Else \quad $s\gets s-U(n'_1,\dotsc,n'_\ell)$
\EndIf
\EndFor
\State \Return $\textsc{Leaf}(P_{m_b}, s)$
\EndFunction
\end{algorithmic}
\caption{Algorithm computing $s$-th leaf of $P$}
\label{alg:leaf}
\end{algorithm}
\begin{proposition}\label{prop:leaf}
For any problem $Q$ that corresponds to a leaf node of the computation tree of a problem $P$ with parameters $n_1,\dotsc,n_\ell$, there exists $s\in\{1,2,\dotsc,U(n_1,\dotsc,n_\ell)\}$ such that $\textsc{Leaf}(P,s)=Q$.
\end{proposition}
\begin{proof}
The proof is an induction on the depth of the computation tree for $P$.
If the computation tree for $P$ consists of a single node, then Algorithm~\ref{alg:leaf} returns $P$.
Assume that the proposition holds for any $P$ with the computation tree of depth at most $d$.
We will consider a problem $P$ with computation tree of depth $d+1$.
Let $i$ be the index of the branching at $P$ that achieves $Q$.
From the induction hypothesis, there exists $s'\in\{1,\dotsc,U(f_1^{b,i}(n_1),\dotsc,f_\ell^{b,i}(n_\ell))\}$ such that $\textsc{Leaf}(P_i, s')=Q$.
Let $s:=s'+\sum_{j=1}^{i-1} U(f_1^{b,j}(n_1),\dotsc,f_\ell^{b,j}(n_\ell))$.
Then, $\textsc{Leaf}(P,s)=Q$.
Here, $s\le \sum_{j=1}^{m_{b}}U(f_1^{b,j}(n_1),\dotsc,f_\ell^{b,j}(n_\ell)) \le U(n_1,\dotsc,n_\ell)$.
\end{proof}
From Proposition~\ref{prop:leaf} and a fact that $\textsc{Leaf}(P,s)$ always returns a problem corresponding to one of the leaf nodes for $P$,
we obtain
\begin{equation*}
a(P) = \bigvee_{s\in\{1,\dotsc,U(n_1,\dotsc,n_\ell)\}} a(\textsc{Leaf}(P,s)).
\end{equation*}
Since the depth of the computation tree for $P$ is at most $\sum_{j=1}^\ell n_j$, the running time of $\textsc{Leaf}(P,s)$ is polynomial with respect to the parameters.
Hence, there is a quantum algorithm computing $a(P)$ with running time $\mathrm{poly}(n_1,\dotsc,n_\ell)\sqrt{U(n_1,\dotsc,n_\ell)}T$.
\end{proof}
For a problem $P$ whose solution is an integer, we can also consider a branching algorithm satisfying $a(P) = \min_{i=1}^{m_{b(P)}} a(P_i)$ for children $P_1,\dotsc,P_{m_{b(p)}}$ of $P$.
In this case, we will call this algorithm MIN-branching algorithm.
Similarly to OR-branching algorithm, we can apply Grover's search to MIN-branching algorithm from Theorem~\ref{thm:min}.

In this paper, we apply Lemma~\ref{lem:or} to Byskov's algorithm in Theorem~\ref{thm:byskov}.
Byskov showed the upper bound $I(n,t)$ satisfying the conditions in Lemma~\ref{lem:or} for the branching algorithm with two parameters $n$ and $t$.
Hence, we can apply Grover's search to Byskov's algorithm in Theorem~\ref{thm:byskov}.
Since $\sum_{t=1}^n I(n,t) = O^*(3^{n/3})$, there is a bounded-error quantum algorithm searching all MISs in time $O^*(3^{n/6})$ as well.

\section{Quantum algorithms for the chromatic number problem}\label{sec:chr}
The overview of the quantum algorithm was described in Section~\ref{subsubsec:algqram}.
The quantum algorithm for Theorem~\ref{thm:main} is shown in Algorithm~\ref{alg:chr}.
\begin{algorithm}[t]
\begin{algorithmic}[1]
\Function{CHR}{$G$}
\If {$G$ is two colorable}
\Return the chromatic number of $G$
\EndIf
\State $\chi[\varnothing]\gets 0$
\For {$S\subseteq V,\,S\ne\varnothing,\,|S|\le \lfloor n/4\rfloor$} (any order consistent with the inclusion relation)
\State $\chi[S] \gets 1 + \min_{I\in\mathrm{MIS}(G[S])}\{\chi[S\setminus I]\}$
\EndFor
\State \Return \textsc{CHR1}($V$)
\EndFunction
\item[]
\Function{CHR1}{$S$}
\State $c\gets |S|$
\For {$t\in\{1,\dotsc,|S|\}$, $s\in\{\max\{\lceil |S|/2\rceil-t,\,1\},\dotsc,\lfloor (|S|-t)/2\rfloor\}$}
  \State $a\gets\min_{I\in\mathrm{MIS}(G[S]),\, |I|=t} \min_{T\subseteq S\setminus I,\, |T|=s} \left(\textsc{CHR2}(T) + \textsc{CHR2}(S\setminus I\setminus T)\right)$
  \State $c\gets\min\{c,a\}$
\EndFor
\State \Return $c+1$
\EndFunction
\item[]
\Function{CHR2}{$S$}
\If {$G[S]$ is two colorable}
\Return the chromatic number of $G[S]$
\EndIf
\State $c\gets |S|$
\For {$t\in\{1,\dotsc,|S|\}$, $s\in\{\max\{\lceil |S|/2\rceil-t,\,1\},\dotsc,\lfloor (|S|-t)/2\rfloor\}$}
  \State $a\gets\min_{I\in\mathrm{MIS}(G[S]),\, |I|=t} \min_{T\subseteq S\setminus I,\, |T|=s} \left(\chi[T] + \chi[S\setminus I\setminus T]\right)$
  \State $c\gets\min\{c,a\}$
\EndFor
\State \Return $c+1$
\EndFunction
\end{algorithmic}
\caption{Algorithm computing the chromatic number of $G$. Grover's search is used for $\min$s.}
\label{alg:chr}
\end{algorithm}
For computing the chromatic number of $G[S]$, when MIS $I$ of size $t$ is chosen, then we have to chose $T\subseteq S\setminus I$ satisfying $|T|\le |S|/2$ and $|S\setminus I\setminus T|\le |S|/2$ as mentioned in Section~\ref{subsubsec:algqram}.
This implies the condition $|S|/2 - t \le |T|\le |S|/2$.
Hence, Algorithm~\ref{alg:chr} computes the chromatic number correctly.
By analyzing the running time of Algorithm~\ref{alg:chr}, we obtain the following theorem.
\begin{theorem}
Algorithm~\ref{alg:chr} computes the chromatic number of $n$-vertex graph with running time
$O^*\bigl((2^{37/35}\allowbreak 3^{3/7}\allowbreak 5^{-9/70}\allowbreak 7^{-5/28})^n\bigr)=O(1.9140^n)$ with bounded error probability.
\end{theorem}
\begin{proof}
The running time of the precomputation is
$O^*\left(\sum_{i=1}^{\lfloor n/4\rfloor} \binom{n}{i} 3^{i/6}\right)
=O^*\left(2^{h(1/4)n} 3^{n/24}\right)
=O(1.8370^n)$.
Let $T_1(n)$ be the running time of $\textsc{CHR1}(V)$ and $T_2(m)$ be the running time of $\textsc{CHR2}(S)$ for $S\subseteq V$ of size $m$.
Then, we obtain
\begin{align}
T_2(m) &= \sum_{t=1}^m \sqrt{I(m,t)} \sum_{s=\max\{\lceil m/2\rceil - t,\, 1\}}^{\lfloor(m-t)/2\rfloor}\sqrt{\binom{m-t}{s}},\label{eq:t2}\\
T_1(n) &= \sum_{t=1}^n \sqrt{I(n,t)} \sum_{s=\max\{\lceil n/2\rceil - t,\, 1\}}^{\lfloor(n-t)/2\rfloor}\sqrt{\binom{n-t}{s}}\left(T_2(s) + T_2(n-t-s)\right)\nonumber\\
&\le \sum_{t=1}^n \sqrt{I(n,t)} \sum_{s=0}^{\min\{\lfloor n/2\rfloor,\, n-t\}}\sqrt{\binom{n-t}{s}}T_2(s)
\label{eq:t1}
\end{align}
by ignoring polynomial factors in $n$.
Here, $T_2(m) \le \sum_{t=1}^m\sqrt{I(m,t)}2^{\frac{m-t}2}$ whose exponent is equal to
$\max_{\delta\in[0,1]} \left\{(E(\delta) + (1-\delta))/2\right\}$.
From Corollary~\ref{cor:I}, it is sufficient to take maximum among $\delta$ being an inverse integer.
Numerical calculation shows that the maximum is given at $\delta=1/5$ and hence $T_2(m) = O^*(80^{m/10})=O(1.5500^m)$.
Hence, the exponent of $T_1(n)$ is equal to
\begin{align*}
\max_{\delta\in[0,1/3],\, \lambda\in[0,1/2]} \left\{\frac12 E(\delta) + \frac12 h\left(\frac{\lambda}{1-\delta}\right)(1-\delta) + \left(\frac1{10}\log 80\right)\lambda\right\}.
\end{align*}
Here, we only consider maximum for $t\le n/3$ since $I(n,t)$ is decreasing with respect to $t$ for $t\ge n/3$, and since the another part $\sum_s \sqrt{\binom{n-t}{s}}T_2(s)$ in~\eqref{eq:t1} is decreasing with respect to $t$.
Numerical calculation shows that the maximum is given at $\delta=1/7, \lambda=1/2$.
Hence, we obtain
\begin{align*}
T_1(n) &= O^*\left(\left(7^{1/14} 2^{h(7/12)3/7} 80^{1/20}\right)^n\right)\\
&=O^*\left((2^{37/35}3^{3/7}5^{-9/70}7^{-5/28})^n\right)=O(1.9140^n).
\qedhere
\end{align*}
\end{proof}
Careful readers may notice that the running time of the precomputation and the main computation are not balanced.
If the quantum algorithm precomputes the chromatic number of induced subgraphs with size at most $(1/4+\epsilon)n$ for some $\epsilon>0$,
the precomputation and the main computation require more and less running time, respectively (we can use Fact~\ref{fa:bcut} for unbalanced $m$).
By optimizing $\epsilon$ such that the both running time are balanced, we may obtain improved running time.
This idea improved the running time of the quantum algorithm for TSP~\cite{ambainis2019quantum}, but does not improve the running time of Algorithm~\ref{alg:chr}.
Equation~\eqref{eq:t2} is dominated by $t=n/5$ and $s=(2/5)n$. 
Equation~\eqref{eq:t1} is dominated by $t=n/7$ and $s=n/2$. 
In order to exclude $s=(2/5)n$ in the summation in~\eqref{eq:t2}, the chromatic number of induced subgraph with size at most $(3/10)n$ must be precomputed.
However, the running time of the precomputation in this case is $\sum_{i=1}^{(3/10)n}\binom{n}{i}3^{i/6}=\Omega(1.9460^n)$.
Hence, the running time of quantum algorithm is not improved.

\section{Quantum algorithms not using QRAM}\label{sec:kcol}
\subsection{Known classical algorithms for $k$-coloring problem}
Beigel and Eppstein showed the fastest known classical algorithm for the graph 3-coloring problem.
\begin{theorem}[Beigel and Eppstein~\cite{beigel20053}]\label{thm:epp}
There is a classical algorithm for the graph 3-coloring problem with running time
$O^*((2^{3/49}3^{4/49}\Lambda^{24/49})^n)=O(1.3289^n)$ where $\Lambda$ denotes the unique real positive root of $x^5-2x-2$.
\end{theorem}

For larger $k$, Byskov showed reduction algorithms from the graph $k$-coloring problem to the graph 3-coloring problem~\cite{byskov2004enumerating}.
Since a graph $G$ is $k$-colorable if and only if there exists an MIS $I$ of size at least $\lceil n/k\rceil$ such that $G[V\setminus I]$ is $(k-1)$-colorable,
we obtain the following reduction algorithm.
\begin{ralgo}[Byskov~\cite{byskov2004enumerating}, Lawler\cite{lawler1976note}]\label{ralgo:1}
For each $t\in\{\lceil n/k\rceil, \lceil n/k\rceil +1,\dotsc, n\}$, enumerate all $t$-MISs.
For each $t$-MIS $I$, the algorithm for the graph $(k-1)$-coloring is performed to $G[V\setminus I]$.
\end{ralgo}
Let $T_k^{(1)}(n)$ be the running time of an algorithm for the graph $k$-coloring problem using Reduction algorithm~\ref{ralgo:1}.
Then, it satisfies
\begin{equation}\label{eq:T1}
T_k^{(1)}(n) = \sum_{t=\lceil n/k\rceil}^n I(n,t) T_{k-1}^{(1)}(n-t)
\end{equation}
for $k\ge 4$ by ignoring a polynomial factor.
By using Reduction algorithm~\ref{ralgo:1} and Theorem~\ref{thm:epp}, Byskov obtained algorithms for the graph 4- and 5-coloring problems with running time $O(1.7504^n)$ and $O(2.1592^n)$, respectively.
Byskov also introduced another reduction algorithm.
Here, we introduce it in a general form.
Since a graph $G$ is $k$-colorable if and only if there exists a subset $S$ of vertices of size at least $\lceil n k' / k\rceil$ such that $G[S]$ is $k'$-colorable and $G[V\setminus S]$ is $(k-k')$-colorable for arbitrary $k'< k$,
we obtain the following reduction algorithm.
\begin{ralgo}[Byskov~\cite{byskov2004enumerating}, Lawler\cite{lawler1976note}]\label{ralgo:2}
Fix $k'\in\{2,3,\dotsc,\lfloor k/2\rfloor\}$.
For each $t\in\{\lceil nk'/k\rceil, \lceil nk'/k\rceil + 1, \dotsc,n\}$, enumerate all subsets of vertices of size $t$.
For each subset $S$ of vertices of size $t$, the algorithms for the graph $k'$- and $(k-k')$-coloring are performed to $G[S]$ and $G[V\setminus S]$, respectively.
\end{ralgo}
Let $T_k^{(2)}(n)$ be the running time of an algorithm for the graph $k$-coloring problem using Reduction algorithm~\ref{ralgo:2}.
Then, it satisfies
\begin{equation}\label{eq:T2}
T_k^{(2)}(n) = \sum_{t=\lceil nk'/k\rceil}^n \binom{n}{t} \left(T_{k'}^{(2)}(t) + T_{k-k'}^{(2)}(n-t)\right)
\end{equation}
for $k\ge 4$ by ignoring a polynomial factor.
Reduction algorithm~\ref{ralgo:2} is a simple generalization of a reduction algorithm in~\cite{byskov2004enumerating}.
For $k=6$, Reduction algorithm~\ref{ralgo:2} with $k'=3$ gives $T_6^{(2)}(n)=O(2.3289^n)$ while $T_6^{(1)}(n)=O(2.5602^n)$~\cite{byskov2004enumerating}.

\subsection{Quantum algorithms not using QRAM}
In this section, we present quantum algorithms not using QRAM, and prove a weaker version of Theorem~\ref{thm:maing} that is valid for $k\le 19$ rather than $k\le 20$.
Theorem~\ref{thm:maing} is obtained by improving quantum algorithms presented in this section.
The improved quantum algorithm and the proof of Theorem~\ref{thm:maing} are shown in Appendix~\ref{apx:20alg}.
F\"urer mentioned that Grover's search can be applied to Beigel and Eppstein's algorithm.

\begin{lemma}[F\"urer~\cite{furer2008solving}]\label{lem:furer}
There is a bounded-error quantum algorithm not using QRAM for the graph 3-coloring problem with running time
$O^*((2^{3/49}3^{4/49}\Lambda^{24/49})^{n/2})=O(1.1528^n)$ where $\Lambda$ denotes the unique real positive root of $x^5-2x-2$.
\end{lemma}
The quantum algorithm in Lemma~\ref{lem:furer} is obtained by application of Lemma~\ref{lem:or} to classical Beigel and Eppstein's algorithm in Theorem~\ref{thm:epp}.
On the other hand, Beigel and Eppstein also showed a very simple randomized algorithm for the graph 3-coloring problem with running time $O^*(2^{n/2})$~\cite[Corollary 1]{beigel20053}.
This randomized algorithm searches one of the $2^{n/2}$ solutions from $2^n$ leaves in a quaternary computation tree of depth $n/2$.
From Theorem~\ref{thm:grover}, we can apply Grover's search to the randomized algorithm and obtain a quantum algorithm with running time $O^*(2^{n/4})$.
The quantum algorithms for Theorem~\ref{thm:maing} reduce the graph $k$-coloring problems to the graph 3-coloring problem.
Theorem~\ref{thm:maing} can be obtained even if this simpler quantum algorithm is used in place of the involved quantum algorithm in Lemma~\ref{lem:furer} although the exponents increase.
By applying Reduction algorithms~\ref{ralgo:1} and \ref{ralgo:2} with Grover's search, the following Theorem is obtained.

\begin{theorem}\label{thm:19}
Assume that there is a polynomial-space bounded-error quantum algorithm not using QRAM for the graph 3-coloring problem with running time $O^*(2^{f^*_3 n})$ for some $f_3^*$.
Then, there is a polynomial-space bounded-error quantum algorithm not using QRAM for the graph $k$-coloring problem with running time $O^*(2^{f^*_k n})$ where
\begin{equation*}
f^*_k := \min
\begin{cases}
\max_{s\in\{3,4,\dotsc,k\}} \{(\log s)/(2s) + (1-1/s)f^*_{k-1}\},\\
\min_{2\le k'\le \lfloor k/2\rfloor}\max_{\delta\in[k'/k, 1]} \left\{h(\delta)/2 + \max\{\delta f^*_{k'},\, (1-\delta)f^*_{k-k'}\}\right\}
\end{cases}
\end{equation*}
for $k\ge 4$.
\end{theorem}
\begin{proof}
We consider quantum algorithm using Reduction algorithms~\ref{ralgo:1} and \ref{ralgo:2} with Grover's search.
Since we can apply Lemma~\ref{lem:or} to Byskov's enumeration algorithm of $t$-MIS in Reduction algorithm~\ref{ralgo:1}, we obtain
\begin{equation*}
T_k^*(n) \le \sum_{t=\lceil n/k\rceil}^n \sqrt{I(n,t)} T_{k-1}^*(n-t).
\end{equation*}
For Reduction algorithm~\ref{ralgo:2}, we can simply apply Grover's search, and obtain
\begin{equation*}
T_k^*(n) \le \sum_{t=\lceil nk'/k\rceil}^n \sqrt{\binom{n}{t}} \left(T_{k'}^*(t) + T_{k-k'}^*(n-t)\right)
\end{equation*}
for any $k'\in\{2,\dotsc,\lfloor k/2\rfloor\}$.
Hence, by choosing the best reduction algorithm, we obtain quantum algorithms with running time $O^*(2^{F^*_k n})$ where
$F^*_3=f^*_3$ and
\begin{equation*}
F^*_k := \min
\begin{cases}
\max_{\delta\in[1/k, 1]} \left\{E(\delta)/2 + (1-\delta)F^*_{k-1}\right\},\\
\min_{2\le k'\le \lfloor k/2\rfloor}\max_{\delta\in[k'/k, 1]} \left\{h(\delta)/2 + \max\{\delta F^*_{k'},\, (1-\delta)F^*_{k-k'}\}\right\}
\end{cases}
\end{equation*}
for $k\ge 4$.
Since $E(\delta)$ is decreasing for $\delta \ge 1/3$, we can assume that $\delta\le 1/3$.
From Corollary~\ref{cor:I}, it is sufficient to take maximum among $\delta=1/s$ for $s\in\{3,4,\dotsc,k\}$.
This proves $F^*_k = f^*_k$.
\end{proof}
By using Theorem~\ref{thm:19} with $f^*_3 = (3 + 4\log 3 + 24 \log \Lambda)/98\le 0.2051$ or $f^*_3 = 1/4$, we obtain Theorem~\ref{thm:maing} for $k\le 19$.
In Table~\ref{tbl:19} in Appendix~\ref{apx:19}, the values of $f^*_k$ and best choices of $k'$ are shown for $f^*_3 = (3 + 4\log 3 + 24 \log \Lambda)/98$.
We summarize the quantum algorithm in Algorithm~\ref{alg:19}.
\begin{algorithm}[t]
\begin{algorithmic}[1]
\Function{COL}{$G$, $k$}
\If {$k\le 2$} \Return the $k$-colorability of $G$ by a polynomial-time algorithm
\EndIf
\If {$k = 3$} \Return the 3-colorability of $G$ by Beigel and Eppstein's algorithm with Grover's search
\EndIf
\If {$k\le 5$}
  \For {$t\in\{\lceil n/k\rceil,\dotsc,n\}$}
    \If {$\bigvee_{I\in\mathrm{MIS}(G), |I|=t} \textsc{COL}(G[V\setminus I], k-1)$} \Return true
    \EndIf
  \EndFor
  \State \Return false
\EndIf
\State Choose $k'$ depending on $k$ from Table~\ref{tbl:19}
\For {$t\in\{\lceil nk'/k\rceil,\dotsc,n\}$}
  \If {$\bigvee_{S\subseteq V, |S| = t} \textsc{COL}(G[S], k') \wedge \textsc{COL}(G[V\setminus S], k-k')$}
    \Return true
  \EndIf
\EndFor
\State \Return false
\EndFunction
\end{algorithmic}
\caption{Algorithm for the $k$-colorability of $G$. Grover's search is used for two $\bigvee$s.}
\label{alg:19}
\end{algorithm}
It is easy to calculate the exponents $f^*_k$ efficiently and precisely.
The details are explained in Appendix~\ref{apx:19}.
Note that if Grover's search are not used in Algorithm~\ref{alg:19}, we obtain classical algorithms with running time $O^*(4^{f^*_kn})$.
By improving Algorithm~\ref{alg:19}, we obtain Theorem~\ref{thm:maing}.
When Reduction algorithm~\ref{ralgo:2} is applied, we can assume that $k-k'$ independent sets in $G[V\setminus S]$ are smaller than $k'$ independent sets in $G[S]$. 
Hence, we can use the average size $|S|/k'$ of $k'$ independent sets in $G[S]$ as an upper bound of independent sets in $G[V\setminus S]$.
This idea reduces the running time of Algorithm~\ref{alg:19} and gives Theorem~\ref{thm:maing}.
The details are shown in Appendix~\ref{apx:20alg}.

\section*{Acknowledgment}
This work was supported by JST PRESTO Grant Number JPMJPR1867 and JSPS KAKENHI Grant Numbers JP17K17711 and JP18H04090.
The authors are grateful to Fran\c{c}ois Le~Gall for helpful discussions.

\bibliographystyle{plain}
\bibliography{biblio}

\appendix
\section{Calculations of the exponents $f^*_k$}\label{apx:19}
\begin{table}[t]
\centering
\caption{Precise values of $f^*_k$.}\label{tbl:19}
\begin{tabular}{|c|c|c|c|}
\hline
$k$& $f^*_k$& $2^{f^*_k}$ & $k'$\\
\hline
3& 0.2050919796& 1.1527598391&\\
4& 0.4038189847& 1.3230054317& 1\\
5& 0.5552479972& 1.4694212030& 1\\
6& 0.6098104848& 1.5260587298& 3\\
7& 0.7233677736& 1.6510316464& 3\\
8& 0.7298058730& 1.6584159226& 4\\
9& 0.8040091395& 1.7459462428& 4\\
10& 0.8297793332& 1.7774134780& 5\\
11& 0.8297793332& 1.7774134780& 5\\
\hline
\end{tabular}
\begin{tabular}{|c|c|c|c|}
\hline
$k$& $f^*_k$& $2^{f^*_k}$ & $k'$\\
\hline
12& 0.8675130685& 1.8245150716& 6\\
13& 0.8873694503& 1.8498001987& 6\\
14& 0.9096459955& 1.8785844800& 6\\
15& 0.9487955413& 1.9302604739& 7\\
16& 0.9487955413& 1.9302604739& 7\\
17& 0.9535113456& 1.9365803294& 8\\
18& 0.9565265484& 1.9406319746& 8\\
19& 0.9713689548& 1.9607001959& 8\\
20& 1.0059831384& 2.0083116140& 8\\
\hline
\end{tabular}
\end{table}
Here, we consider calculations of the exponents $f^*_k$.
This problem was not dealt in~\cite{ambainis2019quantum}.
Fortunately, the calculation of $f^*_k$ is not difficult.
Non-trivial part is the calculation of
\begin{equation*}
\max\left\{\max_{\delta\in[k'/k, 1]} \{h(\delta)/2 + \delta f^*_{k'}\},\, \max_{\delta\in[k'/k, 1]}\{h(\delta)/2 + (1-\delta)f^*_{k-k'}\}\right\}
\end{equation*}
Here, the binary entropy function plus a linear function is a concave function.
Hence, it is sufficient to find a stationary point $\delta^*$.
For maximizing $h(\delta)/2+\delta f^*_{k'}$, we need to find $\delta^*\in[0, 1]$ such that
\begin{equation*}
\frac12\log\frac{1-\delta^*}{\delta^*} + f^*_{k'} = 0.
\end{equation*}
Since the left-hand side is monotonically decreasing, $\delta^*$ can be approximated efficiently by the binary search.
If $\delta^*< k'/k$, then, $\delta=k'/k$ gives the maximum.
The precise values of $f^*_k$ and chosen $k'$ are shown in Table~\ref{tbl:19}.
Here, $k'=1$ means that Reduction algorithm~\ref{ralgo:1} is chosen.

\section{Improved quantum algorithm for $k$-coloring problems not using QRAM}\label{apx:20alg}
\begin{theorem}\label{thm:20}
Assume that there is a polynomial-space bounded-error quantum algorithm not using QRAM for the graph 3-coloring problem with running time $O^*(2^{f^*_3 n})$ for some $f_3^*$.
Then, there is a polynomial-space bounded-error quantum algorithm not using QRAM for the graph $k$-coloring problem with running time $O^*(2^{d^*_k(1) n})$ where
$d^*_3(\mu):=f_3^*$ for all $\mu\in[1/3,1]$ and
\begin{equation*}
d^*_k(\mu) := \min
\begin{cases}
\max_{s\in\{2,3,\dotsc,k\}} \bigl\{(\log s)/(2s) + (1-1/s)d^*_{k-1}(1)\bigr\},\\
\min_{2\le k'\le \lfloor k/2\rfloor}\max_{\delta\in[k'/k, \min\{1,\mu k'\}]} \bigl\{h(\delta)/2 \\
\qquad + \max\{\delta d^*_{k'}(\min\{1,\,\mu/\delta\}),\, (1-\delta)d^*_{k-k'}(\min\{1,\,\delta/(k'(1-\delta))\})\}\bigr\}
\end{cases}
\end{equation*}
for $k\ge 4$ and $\mu\in[1/k,1]$.
\end{theorem}
\begin{proof}
The quantum algorithms are almost same as those in Theorem~\ref{thm:19}.
When we apply Reduction algorithm~\ref{ralgo:2}, we check for all subsets $S$ of vertices of size $t$ whether $G[S]$ is $k'$-colorable and $G[V\setminus S]$ is $(k-k')$-colorable.
Since we assume that $S$ consists of $k'$ largest independent sets in a coloring, we can safely assume that $|S|\ge \lceil nk'/k\rceil$.
At the same time, we can assume that $k-k'$ independent sets in a coloring of $G[V\setminus S]$ have size at most $\lfloor t/k'\rfloor$, which is the average size of independent sets in a coloring of $G[S]$.
This knowledge can be used for reducing the running time of quantum algorithms.
We can consider a partial function that outputs true if the given graph can be partitioned into $k$ independent sets of size at most $u$, outputs false if the given graph is not $k$-colorable,
and outputs either of true or false for other cases.
For computing this partial function, we can restrict the size of $S$ in Reduction algorithm~\ref{ralgo:2}.
Since we only have to consider independent sets of size at most $u$, we can assume that $|S|$ is at most $uk'$ and at least $n-u(k-k')$.
We can assume that $u$ is at least $\lceil n/k\rceil$ since otherwise there is no solution.
Then, $\lceil nk'/k\rceil \ge n-u(k-k')$.
Hence, we can assume that $|S|$ is at least $\lceil nk'/k\rceil$ and at most $\min\{n,uk'\}$.
Let $T_k^*(n, u)$ denote the running time of the quantum algorithm.
Then, we obtain
\begin{align*}
T_k^*(n, u) &\le \sum_{t=\lceil n/k\rceil}^n \sqrt{I(n,t)} T_{k-1}^*(n-t, n-t),\\
T_k^*(n, u) &\le \sum_{t=\lceil nk'/k\rceil}^{\min\{n,u k'\}} \sqrt{\binom{n}{t}} \left(T_{k'}^*(t, \min\{t, u\}) + T_{k-k'}^*(n-t, \min\{n-t,\lfloor t/k'\rfloor\})\right)
\end{align*}
for any $k'\in\{2,3,\dotsc,k\}$.
Note that in Reduction algorithm~\ref{ralgo:1}, we cannot assume that enumerated MIS has size at most $u$ since size of MIS is not restricted.
On the other hand, $T_{k-1}^*(n-t,n-t)$ in the first inequality can be replaced by $T_{k-1}^*(n-t, \min\{t,u\})$.
However, we do not apply this improvement since numerical calculation show that this does not improve the exponents of running time for $k\in\{3,4,\dotsc,21\}\setminus\{13\}$.
By choosing the best reduction algorithms, we obtain $T^*_k(n,\lfloor\mu n\rfloor)=O^*(2^{d^*_k(\mu)n})$.
\end{proof}
By using Theorem~\ref{thm:20} with $f^*_3 = (3 + 4\log 3 + 24 \log \Lambda)/98 \text{ or } 1/4$, we obtain Theorem~\ref{thm:maing}.
Table~\ref{table:20} shows the values of $d^*_k(1)$.
Here, $k'$ may depend on $u$.
However, even if $k'$ is determined only by $k$, the same exponents are obtained.
Chosen $k'$s are the same as those in Table~\ref{tbl:19} except for $k'=7$ for $k=17$.
We summarize the quantum algorithm in Algorithm~\ref{alg:20}.
\begin{algorithm}[t]
\begin{algorithmic}[1]
\Function{COL}{$G$, $k$, $u$}
\If {$k\le 2$} \Return the $k$-colorability of $G$ by a polynomial-time algorithm
\EndIf
\If {$k = 3$} \Return the 3-colorability of $G$ by Beigel and Eppstein's algorithm with Grover's search
\EndIf
\If {$k\le 5$}
  \For {$t\in\{\lceil n/k\rceil,\dotsc,n\}$}
    \If {$\bigvee_{I\in\mathrm{MIS}(G), |I|=t} \textsc{COL}(G[V\setminus I], k-1, |V\setminus I|)$} \Return true
    \EndIf
  \EndFor
  \State \Return false
\EndIf
\State Choose $k'$ that depends on $k$ from Table~\ref{tbl:19}, but $k'=7$ for $k=17$
\For {$t\in\{\lceil nk'/k\rceil,\dotsc,\min\{n,uk'\}\}$}
  \If {$\bigvee_{S\subseteq V, |S| = t} \textsc{COL}(G[S], k', u) \wedge \textsc{COL}(G[V\setminus S], k-k', \lfloor t/k'\rfloor)$}
    \Return true
  \EndIf
\EndFor
\State \Return false
\EndFunction
\end{algorithmic}
\caption{Algorithm for the $k$-colorability of $G$. If $G$ can be partitioned into $k$ independent sets of size at most $u$, then return true. If $G$ is not $k$-colorable, then return false. Otherwise, return either of true or false. Grover's search is used for two $\bigvee$s}
\label{alg:20}
\end{algorithm}
The details of numerical calculation of $d^*_k(1)$ and their precise values are shown in Appendix~\ref{apx:20}.
Finally, we introduce our ideas that failed to improve the running time.
Similarly to Theorem~\ref{thm:20}, we can assume that in Reduction algorithm~\ref{ralgo:2}, size of independent sets in a coloring of $G[S]$ is lower bounded by $(n-t)/(k-k')$, which is the average size of
independent sets in $G[V\setminus S]$.
However, this idea could not improve the exponents in our numerical calculations.
We also tried to use hybrid algorithms of Reduction algorithms~\ref{ralgo:1} and \ref{ralgo:2}.
We introduce a threshold $s$ of size of independent sets. Then, Reduction algorithms~\ref{ralgo:1} and \ref{ralgo:2} are both applied on the assumptions that the size of the largest independent set in a coloring is at least $s+1$ and at most $s$, respectively.
The threshold $s$ is optimized so that running time of Reduction algorithms are balanced.
This algorithm failed to improve the running time as well.

\section{Calculations of the exponents $d^*_k(\mu)$}\label{apx:20}
It is much more difficult to calculate $d^*_k(1)$ than $f^*_k$.
Non-trivial part is the calculations of
\begin{align*}
\max\Bigl\{&\max_{\delta\in[k'/k, \min\{1,\mu k'\}]} \bigl\{h(\delta)/2 + \delta d^*_{k'}(\min\{1,\,\mu/\delta\})\bigr\},\\
&\max_{\delta\in[k'/k, \min\{1,\mu k'\}]} \bigl\{h(\delta)/2 + (1-\delta)d^*_{k-k'}(\min\{1,\,\delta/(k'(1-\delta))\})\}\Bigr\}.
\end{align*}
In the numerical calculations, the maximum for $\delta$ is taken for all $\delta \in \{1/2^{16}, 2/2^{16},\dotsc, 2^{16}/2^{16}\}$.
Computed $d^*_k(\mu)$ is cached and reused.
Then, we obtain $d^*_k(1)$ in Table~\ref{table:20}.

As another heuristic way, we assume that the functions $h(\delta)/2 + \delta d^*_{k'}(\min\{1,\,\mu/\delta\})$
and $h(\delta)/2 + (1-\delta)d^*_{k-k'}(\min\{1,\,\delta/(k'(1-\delta))\})$ are unimodal, which are functions with a single maximal.
On this assumption, we can apply the golden-section search for finding the maximum efficiently.
Obtained approximations of $d^*_k(1)$ are very close to those calculated by the first method.
Hence, we believe that the golden-section search gives very precise approximations of $d^*_k(1)$ which are shown in Table~\ref{tbl:20}.

\begin{table}[t]
\centering
\caption{Precise values of $d^*_k(1)$. For $k\le 13$, $d^*_k(1) = f^*_k$.}\label{tbl:20}
\begin{tabular}{|c|c|c|c|}
\hline
$k$& $d^*_k(1)$& $2^{d^*_k(1)}$ & $k'$\\
\hline
13 & 0.8873694503 & 1.8498001987 & 6\\
14 & 0.8937052065 & 1.8579416667 & 6\\
15 & 0.9487955413 & 1.9302604739 & 7\\
16 & 0.9487955413 & 1.9302604739 & 7\\
17 & 0.9487955413 & 1.9302604739 & 7\\
18 & 0.9535113456 & 1.9365803294 & 8\\
19 & 0.9689936620 & 1.9574747012 & 8\\
20 & 0.9690025400 & 1.9574867472 & 8\\
21 & 1.0086631422 & 2.0120457954 & 9\\
\hline
\end{tabular}
\end{table}

\section{Further improvements}
Theorem~\ref{thm:main} would be improved by the following idea.
When we consider the partition of a graph $G[V\setminus I]$ to $G[T]$ and $G[V\setminus I\setminus T]$ for some $T\subseteq V$ in Algorithm~\ref{alg:chr},
we can assume that independent sets in a coloring of $G[V\setminus I\setminus T]$ have size at least $n-t-2s$ since otherwise there exists more balanced partition.
Hence, we can use algorithms for the graph $\lfloor (n-t-s)/(n-t-2s)\rfloor$-coloring problem in Section~\ref{sec:kcol} for computing the chromatic number of $G[V\setminus I\setminus T]$.

\end{document}